\documentclass{llncs}

\usepackage[utf8]{inputenc}
\usepackage[T1]{fontenc}
\usepackage{amsmath}
\usepackage{amssymb}
\usepackage{xspace}
\usepackage{fullpage}

\newcommand{\cceq}{\mathop{::=}}
\newcommand{\myeq}[1]{\stackrel{\text{#1}}{\Leftrightarrow}}

\renewcommand{\epsilon}{\varepsilon}
\renewcommand{\phi}{\varphi}

\newcommand{\pow}[1]{2^{#1}}
\newcommand{\nats}{\mathbb{N}}

\newcommand{\set}[1]{\{ #1 \}}
\newcommand{\ap}[0]{\mathrm{AP}}

\newcommand{\ltl}{{LTL}\xspace}
\newcommand{\hyltl}{{HyperLTL}\xspace}
\newcommand{\hyctlstar}{{HyperCTL$^*$}\xspace}
\newcommand{\fol}{{FO}$[<]$\xspace}
\newcommand{\fole}{{FO}$[<,\,E]$\xspace}
\newcommand{\hyfol}{{HyperFO}\xspace}
\newcommand{\kltl}{{KLTL}\xspace}

\newcommand{\F}{{\mathbf{F\,}}}
\newcommand{\G}{{\mathbf{G\,}}}
\newcommand{\U}{{\mathbf{\,U\,}}}
\newcommand{\X}{{\mathbf{X\,}}}

\newcommand{\K}{{\mathbf{\,K\,}}}

\newcommand{\suffix}[2]{#1[#2,\infty]}
\newcommand{\var}{\mathcal{V}}

\newcommand{\phiswap}{\phi_{\mathrm{swp}}}
\newcommand{\phiproper}{\phi_{\mathrm{prp}}}
\newcommand{\bal}{\mathrm{bal}}

\newcommand{\el}[0]{\mathrm{E}}
\newcommand{\successor}[0]{\mathrm{Succ}}
\renewcommand{\succ}[0]{\successor}

\newcommand{\model}[1]{\underline{#1}}
\newcommand{\merge}{\mathrm{mrg}}

\newcommand{\fo}[1]{\mathrm{fo}(#1)}
\newcommand{\game}{\mathcal{G}}

\newcommand{\myquot}[1]{``#1''}

\title{The First-Order Logic of Hyperproperties\thanks{This work was partially supported by the German Research Foundation (DFG) under the project SpAGAT (FI 936/2-3) in the Priority Program 1496 ``Reliably Secure Software Systems'' and under the project TriCS (ZI 1516/1-1).}}

\author{Bernd Finkbeiner \and Martin Zimmermann}

\institute{Reactive Systems Group, Saarland University, 66123 Saarbrücken, Germany\\
\email{\{finkbeiner,zimmermann\}@react.uni-saarland.de}
}

\begin{document}

\maketitle


\begin{abstract}
We investigate the logical foundations of
hyperproperties. Hyperproperties generalize trace properties, which
are sets of traces, to sets of sets of traces.  The most prominent
application of hyperproperties is information flow security:
information flow policies characterize the secrecy and integrity of a
system by comparing two or more execution traces, for example by
comparing the observations made by an external observer on execution
traces that result from different values of a secret variable. 

In this
paper, we establish the first connection between temporal logics for
hyperproperties and first-order logic.  Kamp's seminal theorem (in the
formulation due to Gabbay et al.) states that linear-time temporal
logic (LTL) is expressively equivalent to first-order logic over the
natural numbers with order. We introduce first-order logic over sets
of traces and prove that HyperLTL, the extension of LTL to
hyperproperties, is strictly subsumed by this logic. We furthermore
exhibit a fragment that is expressively equivalent to HyperLTL,
thereby establishing Kamp's theorem for hyperproperties.
\end{abstract}


\section{Introduction}
\label{sec_intro}

Linear-time temporal logic
(LTL)~\cite{Pnueli/1977/TheTemporalLogicOfPrograms} is one of the most
commonly used logics in model
checking~\cite{Baier+Katoen/2008/PrinciplesOfModelChecking},
monitoring~\cite{Havelund+Rosu/04/Efficient}, and reactive
synthesis~\cite{DBLP:series/natosec/Finkbeiner16}, and a prime example
for the ``unusal effectiveness of logic in computer
science''~\cite{journals/bsl/HalpernHIKVV01}.  LTL pioneered the idea
that the correctness of computer programs should not just be specified
in terms of a relation between one-time inputs and outputs, but in
terms of the infinite sequences of such interactions captured by
the \emph{execution traces} of the program. The fundamental properties
of the logic, in particular its ultimately periodic model
property~\cite{Sistla:1985:CPL:3828.3837},
and the connection to first-order
logic via Kamp's theorem~\cite{Kamp68}, have been studied extensively
and are covered in various handbook articles and textbooks
(cf.~\cite{Demri+others/2016/Temporal,Thomas:1997:LAL:267871.267878}).

In this paper, we revisit these foundations in light of the recent
trend to consider not only the individual traces of a computer
program, but properties of \emph{sets} of traces,
so-called \emph{hyperproperties}~\cite{Clarkson+Schneider/10/Hyperproperties}.  The motivation for the study of
hyperproperties comes from information flow security.  Information
flow policies characterize the secrecy and integrity of a system by
relating two or more execution traces, for example by comparing the
observations made by an external observer on traces that result from
different values of a secret variable. Such a comparison can obviously
not be expressed as a property of individual traces, but it can be
expressed as a property of the full set of system traces. Beyond
security, hyperproperties also occur naturally in many other
settings, such as the symmetric access to critical resources in
distributed protocols, and Hamming distances between code words in
coding theory~\cite{DBLP:conf/cav/FinkbeinerRS15}.

HyperLTL~\cite{DBLP:conf/post/ClarksonFKMRS14}, the extension of LTL to hyperproperties, uses \emph{trace
quantifiers} and \emph{trace variables} to refer to multiple traces at
the same time. For example, the formula
\begin{equation}
	\forall \pi.\ \forall \pi'.\ \G (a_\pi \leftrightarrow a_{\pi'}) \label{hyperltlexample}
\end{equation} 
expresses that \emph{all} computation traces must \emph{agree} on the
value of the atomic proposition~$a$ at all times.  The extension is
useful: it has been shown that most hyperproperties studied in
the literature can be expressed in HyperLTL~\cite{markusPhD}. There
has also been some success in extending algorithms for model
checking~\cite{DBLP:conf/cav/FinkbeinerRS15},
monitoring~\cite{DBLP:conf/csfw/AgrawalB16}, and
satisfiability~\cite{DBLP:conf/concur/FinkbeinerH16} from LTL to
HyperLTL.  So far, however, we lack a clear understanding of how deeply
the foundations of LTL are affected by the extension. Of particular interest
would be a characterization of the models of the logic.
Are the models of a satisfiable HyperLTL formula still ``simple'' in the sense of the ultimately periodic model theorem of LTL? 

It turns out that the differences between LTL and HyperLTL are
surprisingly profound.  Every satisfiable LTL formula has a model that
is a (single) ultimately periodic trace. Such models are in particular
finite and finitely representable. One might thus conjecture that
a satisfiable HyperLTL formula has a model that consists of a finite
set of traces, or an $\omega$-regular set of traces, or at least \emph{some}
set of ultimately periodic traces. In Section~\ref{sec_models}, we
refute \emph{all} these conjectures.  Some HyperLTL formulas have only
infinite models, some have only non-regular models, and some have only
aperiodic models. We can even encode the prime numbers in HyperLTL! 

Is there some way, then, to characterize the expressive power of HyperLTL?
For LTL, Kamp's seminal theorem~\cite{Kamp68} (in the formulation due to Gabbay et al.~\cite{Gabbayetal80}) states that LTL is expressively equivalent to first-order logic $\fol$ over the natural numbers with order. In order to formulate a corresponding ``Kamp's theorem for HyperLTL,'' we have to decide how to encode sets of traces as relational structures, which also induces the signature of the first-order logic we consider. We chose to use relational structures that consist of disjoint copies of the natural numbers with order, one for each trace. To be able to compare positions on different traces, we add the \emph{equal-level predicate}~$\el$ (cf.~\cite{Thomas09}), which relates the same time points on different traces. The HyperLTL formula (\ref{hyperltlexample}), for example, is equivalent to the $\fole$ formula \[\forall x.\ \forall y.\ \el(x,y) \rightarrow (P_a(x)\leftrightarrow P_a(y)).\]

In Section~\ref{sec_fo}, we show that $\fole$ is \emph{strictly more expressive} than HyperLTL, i.e., every HyperLTL formula can be translated into an equivalent $\fole$ formula, but there exist $\fole$ formulas that cannot be translated to HyperLTL. Intuitively, $\fole$ can express requirements which relate at some point in time an \emph{unbounded} number of traces, which is not possible in HyperLTL. To obtain a fragment of $\fole$ that is expressively equivalent to HyperLTL, we must rule out such properties. We consider the fragment where the quantifiers either refer to initial positions or are guarded by a constraint that ensures that the new position is on a trace identified by an initial position chosen earlier. In this way, a formula can only express properties of the bounded number of traces selected by the quantification of initial positions. We call this fragment \hyfol, the \emph{first-order logic of hyperproperties}. Theorem~\ref{theo_Hyperkamp}, the main result of the paper, then shows that HyperLTL and \hyfol are indeed expressively equivalent, and thus proves that Kamp's correspondence 
between temporal logic and first-order logic also holds for hyperproperties.

All proofs omitted due to space restrictions can be found in the full version~\cite{DBLP:journals/corr/FinkbeinerZ16}.


\section{HyperLTL}
\label{sec_hyperltl}

Fix a finite set~$\ap$ of atomic propositions. A trace over $\ap$ is a map $t \colon \nats \rightarrow \pow{\ap}$, denoted by $t(0)t(1)t(2) \cdots$. The set of all traces over $\ap$ is denoted by $(\pow{\ap})^\omega$. The projection of $t$ to $\ap'$ is the trace $(t(0) \cap \ap') (t(1) \cap \ap') (t(2) \cap \ap') \cdots $ over $\ap'$. A trace~$t$ is ultimately periodic, if $t = t_0 \cdot t_1^\omega$ for some $t_0,t_1 \in (\pow{\ap})^+$, i.e., there are $s,p>0$ with $t(n) = t(n+p)$ for all $n \ge s$. A set~$T$ of traces is ultimately periodic, if every trace in $T$ is ultimately periodic.

The formulas of \hyltl are given by the grammar 
\begin{align*}
\phi & {} \cceq {}  \exists \pi.\ \phi \mid \forall \pi.\ \phi \mid \psi \\
\psi & {}  \cceq {}  a_\pi \mid \neg \psi \mid \psi \vee \psi \mid \X \psi \mid \psi \U \psi 
\end{align*}
where $a$ ranges over atomic propositions in $\ap$ and where $\pi$ ranges over a given countable set~$\var$ of \emph{trace variables}. Conjunction, implication, equivalence, and exclusive disjunction~$\oplus$ as well as the temporal operators eventually~$\F$ and always~$\G$ are derived as usual. A sentence is a closed formula, i.e., the formula has no free trace variables.

The semantics of \hyltl is defined with respect to a trace assignment, a partial mapping~$\Pi \colon \var \rightarrow (\pow{\ap})^\omega$. The assignment with empty domain is denoted by $\Pi_\emptyset$. Given a trace assignment~$\Pi$, a trace variable~$\pi$, and a trace~$t$ we denote by $\Pi[\pi \rightarrow t]$ the assignment that coincides with $\Pi$ everywhere but at $\pi$, which is mapped to $t$. Furthermore, $\suffix{\Pi}{j}$ denotes the assignment mapping every $\pi$ in $\Pi$'s domain to $\Pi(\pi)(j)\Pi(\pi)(j+1)\Pi(\pi)(j+2) \cdots $.

For sets~$T$ of traces and trace-assignments~$\Pi$ we define 
\begin{itemize}
	\item $(T, \Pi) \models a_\pi$, if $a \in \Pi(\pi)(0)$,
	\item $(T, \Pi) \models \neg \psi$, if $(T, \Pi) \not\models \psi$,
	\item $(T, \Pi) \models \psi_1 \vee \psi_2 $, if $(T, \Pi) \models \psi_1$ or $(T, \Pi) \models \psi_2$,
	\item $(T, \Pi) \models \X \psi$, if $(T,\suffix{\Pi}{1}) \models \psi$,
	\item $(T, \Pi) \models \psi_1 \U \psi_2$, if there is a $j \ge 0$ such that $(T,\suffix{\Pi}{j}) \models \psi_2$ and for all $0 \le j' < j$: $(T,\suffix{\Pi}{j'}) \models \psi_1$, 
	\item $(T, \Pi) \models \exists \pi.\ \phi$, if there is a trace~$t \in T$ such that $(T,\Pi[\pi \rightarrow t]) \models \psi$, and 
	\item $(T, \Pi) \models \forall \pi.\ \phi$, if for all traces~$t \in T$: $(T,\Pi[\pi \rightarrow t]) \models \psi$. 
\end{itemize}

We say that $T$ satisfies a sentence~$\phi$, if $(T, \Pi_\emptyset) \models \phi$. In this case, we write $T \models \phi$ and say that $T$ is a model of $\phi$. Although \hyltl sentences are required to be in prenex normal form, they are closed under boolean combinations, which can easily be seen by transforming such formulas into prenex normal form.


\section{The Models of HyperLTL}
\label{sec_models}

Every satisfiable \ltl formula has an ultimately periodic model, i.e., a particularly simple model: It is trivially finite (and finitely represented) and forms an $\omega$-regular language. An obvious question is whether every satisfiable \hyltl sentence has a simple model, too. Various notions of simplicity could be considered here, e.g., cardinality based ones,  being $\omega$-regular, or being ultimately periodic, which all extend the notion of simplicity for the \ltl case. In this section, we refute all these possibilities: We show that \hyltl models have to be in general infinite, might necessarily be non-regular, and may necessarily be aperiodic.


\subsection{No Finite Models}
\label{subsec_smallmodels}

Our first result shows that \hyltl does not have the finite model property  (in the sense that every satisfiable sentence is satisfied by a finite set of traces). The proof is a straightforward encoding of an infinite set of traces that appears again in the following proofs.

\begin{theorem}
\label{thm_finitenotenough}
There is a satisfiable \hyltl sentence that is not satisfied by any finite set of traces.
\end{theorem}

\begin{proof}
Consider the conjunction~$\phi$ of the following formulas  over~$\ap = \set{a}$:
\begin{itemize}

\item $\forall \pi.\ (\neg a_\pi) \U (a_\pi \wedge \X\G\neg a_\pi )$: on every trace there is exactly one occurrence of $a$.

\item $\exists \pi.\ a_\pi$: there is a trace where $a$ holds true in the first position.

\item $\forall \pi.\ \exists \pi'.\ \F (a_\pi \wedge \X a_{\pi'})$: for every trace, say where $a$ holds at position~$n$ (assuming the first conjunct is satisfied), there is another trace where $a$ holds at position~$n+1$.

\end{itemize}
It is straightforward to verify that $\phi$ is satisfied by the infinite set~$T = \set{ \emptyset^n \cdot \set{a} \cdot \emptyset^\omega \mid n \ge 0 }$ and an induction over $n$ shows that every model has to contain $T$. Here, one uses the first and second conjunct in the induction start and the first and third conjunct in the induction step. Actually, the first conjunct then implies that $T$ is the only model of $\phi$.
\end{proof}

Next, we complement the lower bound with a matching upper bound.

\begin{theorem}
\label{thm_countableenough}
Every satisfiable \hyltl sentence has a countable model.
\end{theorem}

\begin{proof}
Let $\phi$ be a satisfiable \hyltl sentence and let $T$ be a model. If $T$ is countable, then we are done. Thus, assume $T$ is uncountable and thus in particular non-empty. Furthermore, we assume w.l.o.g.\ $\phi = \forall \pi_0.\ \exists \pi_0'. \cdots \forall \pi_k.\ \exists \pi_k'.\ \psi $ with quantifier-free $\psi$. 

As $T$ is a model of $\phi$, there is a Skolem function $f_i \colon T^{i} \rightarrow T$ for every $i \le k $ satisfying the following property: $(T, \Pi) \models \psi$ for every trace assignment~$\Pi$ that maps each $\pi_i$ to some arbitrary~$t_i \in T$ and every $\pi_i'$ to $f_i(t_0, \ldots, t_i)$. Note that the relation~$(T, \Pi) \models \psi$ does only depend on $\Pi$ and $\psi$, but not on $T$, as $\psi$ is quantifier-free.

Given a subset $S \subseteq T$ and a Skolem function~$f_i$ we define 
\[f_i(S) = \set{ f_i(t_0, \ldots , t_{i}) \mid t_0, \ldots , t_{i} \in S }.\]
Now, fix some $t \in T$. Define $S_{0} = \set{ t }$ and $ S_{ n+1 } = S_n \cup \bigcup_{i=0}^k f_i(S_n)$ for every $n$, and $S = \bigcup_{n \ge 0}S_n$. The limit stage~$S$ is closed under applying the Skolem functions, i.e., if $t_0, \ldots, t_i \in S$, then $f_i(t_0, \ldots, t_i) \in S$. Also, every stage~$S_n$ is finite by a straightforward induction, hence $S$ is countable. We conclude the proof by showing that $S$ is a model of $\phi$. 

Every trace assignment~$\Pi$ mapping $\pi_i$ to some $t_i \in S$ and every $\pi_i'$ to $f_i(t_0, \ldots, t_i) \in S $ satisfies $(T, \Pi) \models \psi$, as argued above. Also, as argued above, this is independent of $T$ due to $\psi$ being quantifier-free. Hence, we obtain $(S, \Pi) \models \psi$. Finally, a simple induction over the quantifier prefix shows $(S, \Pi_\emptyset) \models \phi$, i.e., $S$ is indeed a model of $\phi$.
\end{proof}


\subsection{No Regular Models}
\label{subsec_regmodels}

The construction presented in the proof of Theorem~\ref{thm_finitenotenough}, which pushes a single occurrence of the proposition~$a$ through the traces to enforce the set $\set{ \emptyset^n \cdot \set{a} \cdot \emptyset^\omega \mid n \ge 0 }$ is reused to prove the main result of this subsection. We combine this construction with an inductive swapping construction to show that \hyltl sentences do not necessarily have $\omega$-regular models. To illustrate the swapping, consider the following finite traces:
\begin{align*}
&t_0 = \set{a} \cdot \emptyset \cdot \set{a} \cdot \emptyset \cdot \set{a} \cdot \emptyset \qquad\qquad  t_2 = \set{a} \cdot \set{a} \cdot \emptyset \cdot \set{a} \cdot \emptyset \cdot \emptyset \qquad\qquad\\
&t_1 = \set{a} \cdot \set{a} \cdot \emptyset  \cdot \emptyset \cdot \set{a} \cdot \emptyset \qquad \qquad t_3 = \set{a} \cdot \set{a} \cdot \set{a} \cdot \emptyset \cdot \emptyset \cdot \emptyset
\end{align*}
The trace $t_1$ is obtained from $t_0$ by swapping the first occurrence of $\emptyset$ one position to the right (a swap may only occur between adjacent positions, one where $a$ holds and one where it does not). Furthermore, with two more swaps, one turns $t_1$ into $t_2$ and $t_2$ into $t_3$.

Our following proof is based on the following three observations: (1) In an alternating sequence of even length such as $t_1$, the number of positions where $a$ holds and where $a$ does not hold is equal. Such a sequence is expressible in (Hyper)LTL. (2) A swap does not change this equality and can be formalized in \hyltl. (3) Thus, if all occurrences of $\set{a}$ are swapped to the beginning, then the trace has the form $\set{a}^n \cdot \emptyset^n$ for some $n$. Hence, if we start with all alternating sequences as in $t_0$, then we end up with the non-regular language~$\set{\set{a}^n \cdot \emptyset^n \mid n>0}$. 

\begin{theorem}
\label{thm_regularnotenough}
There is a satisfiable \hyltl sentence that is not satisfied by any $\omega$-regular set of traces.
\end{theorem}

\begin{proof}
Consider the conjunction~$\phi$ of the formulas~$\phi_i$, $i \in \set{1, \ldots, 8}$ over~$\ap = \set{a,b,1,2,\dagger}$.

\begin{itemize}
	\item $\phi_1 = \forall \pi.\ (1_\pi \oplus 2_\pi ) \wedge \neg \dagger_\pi  \wedge \neg \dagger_\pi  \U \G(\dagger_\pi  \wedge \neg a_\pi)$.
\end{itemize}
Every trace from a set of traces satisfying $\phi_1$ either satisfies $1$ or $2$ at the first position. Consequently, we speak of traces of type~$i$ for $i \in \set{1,2}$. Also, on every such trace the truth value of $\dagger$ changes exactly once, from false to true, after being false at least at the first position. In the following, we are only interested in the unique maximal prefix of a trace where $\dagger$ does not hold, which we call the \emph{window} of the trace. Note that $a$ may only hold in the window of a trace. Considering windows essentially turns infinite traces into finite ones.

The balance~$\bal(t)$ of a trace~$t$ is the absolute value of the difference between the number of window positions where $a$ holds and the number of those where $a$ does not hold, i.e.,
\[\bal(t) = |\, |\set{n \mid a \in t(n) \text{ and } \dagger\notin t(n) }| - |\set{n \mid a \notin t(n) \text{ and } \dagger\notin t(n) }|\, |.\]

\begin{itemize}
	\item $\phi_2 = \forall \pi.\ 1_\pi  \rightarrow (a_\pi \wedge \G( a_\pi \rightarrow \X \neg a_\pi \wedge \X\neg \dagger_\pi \wedge \X\X(a_\pi \vee \dagger_\pi) ))$
	
	\item $\phi_3 = \exists \pi.\ 1_\pi \wedge a_\pi \wedge \X\X \dagger_\pi$
	
	\item $\phi_4 = \forall \pi.\ \exists \pi'.\ 1_\pi \rightarrow (1_{\pi'} \wedge \F( \neg\dagger_{\pi} \wedge \X\dagger_\pi \wedge \X\X\neg\dagger_{\pi'} \wedge \X\X\X\dagger_{\pi'} )) $
\end{itemize}
If $\phi_1 \wedge \cdots \wedge \phi_4$ is satisfied by a set of traces, then the projection to $\set{a}$ of the window of every type~$1$ trace has the form $(\set{a}\cdot\emptyset)^n$ for some $n>0$, due to $\phi_2$. In particular, every type~$1$ trace has balance zero. Furthermore, due to $\phi_3$ and $\phi_4$, there is a trace with such a window for every $n>0$. 

\begin{itemize}
	\item $\phi_5 = \forall \pi.\ 2_\pi \rightarrow b_\pi \wedge b_\pi\U\G\neg b_\pi$
\end{itemize}
Finally, $\phi_5$ requires every type~$2$ trace to have a prefix where $b$ holds true, after which it never holds true again. The length of this prefix is the \emph{rank} of the trace, which is finite.

The next formula implements the swapping process. Each swap has to decrease the rank until a type~$1$ trace is reached. This rules out models satisfying the formulas by cyclic swaps.

\begin{itemize}
	\item $\phi_6 = \forall \pi.\ \exists \pi'.\ 2_\pi \rightarrow  (\F (\dagger_{\pi} \wedge \dagger_{\pi'} \wedge \X \neg \dagger_{\pi} \wedge \X\neg\dagger_{\pi'})) \wedge \phiswap(\pi,\pi') \wedge [$ 
	
	$\phantom{\phi_6 = \forall \pi.\ \exists \pi'.\ } (1_{\pi'} \wedge b_{\pi} \wedge \X \neg b_{\pi})\, \vee $
	
	$\phantom{\phi_6 = \forall \pi.\ \exists \pi'.\ } (2_{\pi'} \wedge \F(b_{\pi'} \wedge \X \neg b_{\pi'} \wedge \X b_{\pi} \wedge \X\X \neg b_{\pi})) ]$
	
\end{itemize}
where
\[\phiswap(\pi,\pi') = (a_{\pi}\leftrightarrow a_{\pi'})\U ((a_{\pi} \oplus \X a_{\pi}) \wedge (a_{\pi'} \oplus\X a_{\pi'}) \wedge (a_{\pi} \oplus a_{\pi'}) \wedge \X\X \G( a_{\pi} \leftrightarrow a_{\pi'})) . \]
Intuitively, this formula requires for every trace~$t$ of type~$2$ the existence of a trace~$t'$ of the same window length and where the difference in the truth values of $a$ in $t$ and $t'$ is only a single swap at adjacent positions (first line). Furthermore, if $t$ has rank one, then $t'$ has to be of type~$1$ (line two); otherwise, if $t$ has rank~$r > 1$, then $t'$ has to be of type~$2$ and has to have rank~$r-1$ (line three). Thus, the rank is an upper bound on the number of swaps that can be executed before a trace of type~$1$ is reached.

An induction over the rank of type~$2$ traces shows that every such trace has balance zero, as a swap as formalized by $\phiswap$ does not change the balance. 

\begin{itemize}
	\item $\phi_7 = \exists \pi.\ 2_\pi \wedge a_\pi$
	\item $\phi_8 = \forall \pi.\ \exists \pi'.\ 2_\pi \rightarrow (2_{\pi'} \wedge (a_\pi \wedge a_{\pi'}) \U(\G\neg a_\pi \wedge a_{\pi'} \wedge \X \G\neg a_{\pi'} ))$
\end{itemize}
The last two formulas imply for every $n>0$ the existence of a trace of type~$2$ which has a prefix where $a$ holds true at exactly the first $n$ positions, after which it never holds true again. Due to the balance of type~$2$ traces being zero (assuming all previous formulas are satisfied), the projection to $\set{a}$ of the window of such a trace has the form~$\set{a}^n \cdot \emptyset^n$.

Now, towards a contradiction, assume that $T \models \phi$ for some $\omega$-regular $T$. It follows from the observations made above that projecting $T$ to $\set{a,\dagger}$ and intersecting it with the $\omega$-regular language~$\set{a}^*\cdot\emptyset^*\cdot\set{\dagger}^\omega$ results in the language~$\set{\set{a}^n\cdot\emptyset^n\cdot\set{\dagger}^\omega \mid n >0}$, which is not $\omega$-regular. This yields the desired contradiction. 

To conclude, it suffices to remark that $\phi$ is satisfied by taking the union of the set of all required type~$1$ traces and of the set of all type~$2$ traces with finite window length, balance zero, and with rank equal to the number of swaps necessary to reach a type~$1$ trace. 
\end{proof}

Note that this result can be strengthened by starting with type~$1$ traces of the form $(\emptyset \cdot \set{a} \cdot \set{a'} \cdot \set{a,a'})^+ \set{\dagger}^\omega$ for some fresh proposition~$a'$ and then modify the swap operation to obtain sequences of the form $\emptyset^n \cdot \set{a}^n \cdot \set{a'}^n \cdot \set{a,a'}^n\set{\dagger}^\omega$. These form, when ranging over all $n$, a non-$\omega$-contextfree language (see \cite{DBLP:journals/jcss/CohenG77} for a formal definition of these languages). Thus, not every \hyltl sentence has an $\omega$-contextfree model.

\begin{theorem}
\label{thm_contextfreenotenough}
There is a satisfiable \hyltl sentence that is not satisfied by any $\omega$-contextfree set of traces. 
\end{theorem}

It is an interesting question to find a non-trivial class of languages that is rich enough for every satisfiable \hyltl sentence to be satisfied by a model from this class. 


\subsection{No Periodic Models}
\label{subsec_models}

Next, we extend the techniques developed in the previous two subsections to show our final result on the complexity of \hyltl models: although every \ltl formula has an ultimately periodic model, one can construct a \hyltl sentence without ultimately periodic models. 

\begin{theorem}
\label{thm_noultperiodic}
There is a satisfiable \hyltl sentence that is not satisfied by any set of ultimately periodic traces.
\end{theorem}

\begin{proof}
A trace~$t$ is \emph{not} ultimately periodic, if for every $s,p>0$ there is an $n \ge s$ with $t(n) \neq t(n +p)$. In the following, we construct auxiliary traces that allow us to express this property in \hyltl. The main difficulty is to construct traces of the form $(\set{b}^p \cdot \emptyset^p)^\omega$ for every $p$, to implement the quantification of the period length~$p$. 

We construct a sentence~$\phi$ over $\ap = \set{a, b , 1, 2, \$}$ with the desired properties, which is a conjunction of several subformulas. The first conjunct requires every trace in a model of $\phi$ to have exactly one occurrence of the proposition~$a$. If it holds at position~$n$, then we refer to $n+1$ as the \emph{characteristic} of the trace (recall that a trace starts at position~$0$).

As in the proof of Theorem~\ref{thm_regularnotenough}, we have two special types of traces in models of $\phi$, which are identified by either $1$ or $2$ holding true at the first position of every trace, but there might be other traces as well.
Type~$1$ traces are of the form $\emptyset^c\cdot\set{a}\cdot\emptyset^\omega$ for $c \ge 0$. As in the proof of Theorem~\ref{thm_finitenotenough}, one can construct a conjunct that requires the models of $\phi$ to contain a type~$1$ trace for every such $c$, but no other traces of type~$1$. 

The projection to $\set{b}$ of a trace $t$ of type~$2$ is a suffix of $(\set{b}^c\cdot \emptyset^c)^\omega $, where $c$ is the characteristic of $t$. We claim that one can construct a conjunct of $\phi$ that requires all models of $\phi$ to contain all these type~$2$ traces, i.e., all possible suffixes for every $c>0$. This is achieved by formalizing the following properties in \hyltl:

\begin{enumerate}
	
	\item Every type~$2$ trace has infinitely many positions where $b$ holds and infinitely many positions where $b$ does not hold. A block of such a trace is a maximal infix whose positions coincide on their truth values of $b$, i.e., either $b$ holds at every position of the infix, but not at the last one before the infix (if it exists) and not at the first position after the infix or $b$ does not hold at every position of the infix, but at the last one before it (if it exists) and at the first position after it.
	
	\item For every type~$1$ trace there is at least one type~$2$ trace of the same characteristic.
	
	\item The length of the first block of every type~$2$ trace is not larger than its characteristic.
	
	\item If a block ends at the unique position of a type~$2$ trace where its $a$ holds, then it has to be the first block. 

	\item For every type~$2$ trace there is another one of the same characteristic that is obtained by shifting the truth values of $b$ one position to the left. 
	
\end{enumerate}
	
Assume a set~$T$ of traces satisfies all these properties and assume there is a type~$2$ trace~$t \in T$ whose projection to $\set{b}$ is not a suffix of $(\set{b}^c\cdot \emptyset^c)^\omega $, where $c$ is the characteristic of $t$. The length of its first block is bounded by $c$, due to the third property. Thus, there has to be a non-first block whose length~$\ell$ is not equal to $c$. If $\ell > c$, we can use the fifth property to shift this block to the left until we obtain a type~$2$ trace of characteristic~$c$ in $T$ whose first block has the same length~$\ell$. This trace violates the third property. If $\ell < c$, then we can again shift this block to the left until we obtain a trace in $T$ of characteristic $c$ that has a block of length~$\ell$ that ends at the unique position where $a$ holds. Due to $\ell < c$, this cannot be the first block, i.e., we have derived a contradiction to the fourth property. 

On the other hand, for every $c>0$, there is a some type~$2$ trace of characteristic~$c$ in $T$. As shown above, its projection to $\set{b}$ is a suffix of $(\set{b}^c\cdot \emptyset^c)^\omega $. Thus, applying the left-shift operation $2c-1$ times yields all possible suffixes of $(\set{b}^c\cdot \emptyset^c)^\omega $. Thus, $T$ does indeed contain all possible type~$2$ traces, if it satisfies the formulas described above.

Recall that we have to express the following property: there is a trace~$t$ such that for every $s,p>0$ there is an $n \ge s$ with $t(n) \neq t(n +p)$. To this end, we first existentially quantify a trace~$\pi$ (the supposedly non-ultimately periodic one). Then, we universally quantify two type~$1$ traces~$\pi_s$ and $\pi_p$ (thereby fixing $s$ and $p$ as the characteristics of $\pi_s$ and $\pi_p$). Thus, it remains to state that $\pi$ has two positions~$n$ and $n'$ satisfying $s \le n < n' = n+p$ such that the truth value of $\$$ differs at these positions. To this end, we need another trace~$\pi_p'$ of the same characteristic~$p$ as $\pi_p$ so that a block of $\pi_p'$ starts at position $n$, which allows to determine $n' = n+p$ by just advancing to the end of the block starting at $n$.  

Formally, consider the following statement: there is a trace~$\pi$ such that for all type~$1$ traces~$\pi_s$ and $\pi_p$ (here, we quantify over $s$ and $p$) there is a type~$2$ trace $\pi_p'$ that has the same characteristic as $\pi_p$ such that the following is true: there is a position~$n$ no earlier than the one where $a$ holds in $\pi_s$ such that
\begin{itemize}
	\item the truth value of $b$ in $\pi_p'$ differs at positions~$n-1$ and $n$ (i.e., a block begins at $n$), and
	\item the atomic proposition~$\$$ holds at $n$ in $\pi$ and not at $n'$ in $\pi$ or vice versa, where $n' >n$ is the smallest position such that the truth value of $b$ in $\pi_p'$ differs at $n'-1$ and $n'$ (i.e., the next block begins at position~$n'$), which implies $n' = n + p$. 
\end{itemize}
The formalization of this statement in \hyltl is the final conjunct of $\phi$. Hence, $\phi$ has no models that contain an ultimately periodic trace. 

Finally, $\phi$ is satisfied by all models that contain all possible type~$1$ and all possible type~$2$ traces as well as at least one trace that is not ultimately periodic when projected to $\set{\$}$.\end{proof}

Note that the type~$1$ and type~$2$ traces above are ultimately periodic, i.e., although we have formalized the existence of a single non-ultimately periodic trace, the model always has ultimately periodic ones as well. By slightly extending the construction, one can even construct a satisfiable sentence whose models contain not a single ultimately periodic trace. To this end, one requires that every trace (in particular the type~$1$ and type~$2$ traces) is non-ultimately periodic, witnessed by the proposition~$\$ $ as above. 

\begin{theorem}
\label{thm_noultperiodic_strong}
There is a satisfiable \hyltl sentence that is not satisfied by any set of traces that contains an ultimately periodic trace.
\end{theorem}

As a final note on the expressiveness of \hyltl we show how to encode the prime numbers. Let type~$1$ and type~$2$ traces be axiomatized as in the proof of Theorem~\ref{thm_noultperiodic}. Recall projecting a type~$2$ trace to $\set{b}$ yields a suffix of $(\set{b}^c\cdot \emptyset^c)^\omega $, where $c>0$ is the trace's characteristic. We say that such a trace is \emph{proper}, if its projection equal to $(\set{b}^c\cdot \emptyset^c)^\omega $. Being proper can be expressed in \hyltl, say by the formula~$\phiproper(\pi)$ with a single free variable, relying on the fact that the only occurrence of $a$ induces the characteristic~$c$. Also, we add a new atomic proposition~$'$ to $\ap$ to encode the prime numbers as follows: the proposition~$'$ holds at the first position of a type~$1$ trace of characteristic~$c$ if, and only if, $c$ is a prime number. 

Now, consider the following formula, which we add as a new conjunct to the axiomatization of type~$1$ and type~$2$ traces:
\begin{align*}
& \forall \pi^1.\ \forall \pi^2.\ (1_{\pi^1} \wedge \ '_{\pi^1} \wedge \phiproper(\pi^2) \rightarrow \neg\psi(\pi^1,\pi^2)) \, \wedge\\
 & \forall\pi^1.\ \exists \pi^2.\ (1_{\pi^1} \wedge \neg\ '_{\pi^1} \rightarrow \phiproper(\pi^2) \wedge \psi(\pi^1,\pi^2))
\end{align*}
Here, the formula~$\psi(\pi^1,\pi^2)$ expresses that the single~$a$ in $\pi^1$ appears at the end of a non-first block in $\pi^2$ and that the characteristic of $\pi^2$ is strictly greater than one. Thus, $\psi(\pi^1, \pi^2)$ holds if, and only if, the characteristic of $\pi^2$ is a non-trivial divisor of the characteristic of $\pi^1$. Thus, the first conjunct expresses that a type~$1$ trace of characteristic~$c>1$ may only have a $'$ at the first position, if $c$ has only trivial divisors, i.e., if $c$ is prime. Similarly, the second conjunct expresses that a type~$1$ trace of characteristic~$c>1$ may only not have a $'$ at the first position, if $c$ has a non-trivial divisor, i.e., if $c$ is not prime. Thus, by additionally hardcoding that $1$ is not a prime, one obtains a formula~$\phi$ such that every model~$T$ of $\phi$ encodes the primes as follows: $c$ is prime if, and only if, there is a type~$1$ trace of characteristic~$c$ in $T$ with $'$ holding true at its first position. 


\section{First-order Logic for Hyperproperties}
\label{sec_fo}

Kamp's seminal theorem~\cite{Kamp68} states that Linear Temporal Logic with the until-operator~$\U$ and its dual past-time operator \myquot{since} is expressively equivalent to first-order logic over the integers with order, \fol for short. Later, Gabbay et al.~\cite{Gabbayetal80} proved that \ltl as introduced here (i.e., exclusively with future-operators) is expressively equivalent to first-order logic over the \emph{natural numbers} with order. More formally, one considers relational structures of the form $(\nats, <, (P_a)_{a\in\ap})$ where $<$ is the natural ordering of $\nats$ and each $P_a$ is a subset of $\nats$. There is a bijection mapping a trace~$t$ over $\ap$ to such a structure~$\model{t}$. Furthermore, \fol is first-order logic\footnote{We assume familiarity with the syntax and semantics of first-order logic. See, e.g.,~\cite{EbbinghausFlumThomas94}, for an introduction to the topic.} over the signature~$\set{<} \cup \set{ P_a \mid a \in \ap}$ with equality.
The result of Gabbay et al.\ follows from the existence of the following effective translations: (1)~For every \ltl formula~$\phi$ there is an \fol sentence $\phi'$ such that for all traces~$t$: $t \models \phi$ if, and only if, $\model{t} \models \phi'$. (2)~For every \fol sentence $\phi$ there is an \ltl formula $\phi'$ such that for all traces~$t$: $\model{t} \models \phi$ if, and only if, $t \models \phi'$. 

In this section, we investigate whether there is a first-order logic that is expressively equivalent to \hyltl. The first decision to take is how to represent a set of traces as a relational structure. The natural approach is to take disjoint copies of the natural numbers, one for each trace and label them accordingly. Positions on these traces can be compared using the order. To be able to compare different traces, we additionally introduce a (commutative) equal-level predicate~$\el$, which relates the same time points on different traces.

Formally, given a set~$T \subseteq (\pow{\ap})^\omega$ of traces over $\ap$, we define the relational structure~$\model{T} = (T \times \nats, <^{\model{T} }, \el^{\model{T} }, (P_a^{\model{T} })_{a \in \ap})$ with
\begin{itemize}
	\item $<^{\model{T} } = \set{((t,n),(t, n')) \mid t \in T \text{ and } n < n' \in \nats }$,
	\item $\el^{\model{T} } = \set{ ((t,n),(t',n)) \mid t,t' \in T \text{ and } n \in \nats}$, and
	\item $P_a^{\model{T} } = \set{ (t,n) \mid a \in t(n) }$.
\end{itemize}

We consider first-order logic over the signature~$\set{<,\el } \cup \set{ P_a \mid a \in \ap}$, i.e., with atomic formulas~$x = y$, $x < y$, $\el(x, y)$, and $P_a(x)$ for $a \in \ap$, and disjunction, conjunction, negation, and existential and universal quantification over elements. We denote this logic by $\fole$. 
We use the shorthand~$x \le y$ for $x < y \vee x = y$ and freely use terms like $x \le y < z$ with the obvious meaning. A sentence is a closed formula, i.e., every occurrence of a variable is in the scope of a quantifier binding this variable. We write $\phi(x_0, \ldots, x_n)$ to denote that the free variables of the formula~$\phi$ are among $x_0, \ldots, x_n$. 

\begin{example}\hfill
\label{example_hyfolsyntacticsugar}
\begin{enumerate}
	\item The formula~$\succ(x,y) = x < y \wedge \neg \exists z.\ x < z  < y$ expresses that $y$ is the direct successor of $x$ on some trace. 
	\item The formula~$\min(x) = \neg \exists y.\ \succ(y,x)$ expresses that $x$ is the first position of a trace. 
\end{enumerate}	
\end{example}

Our first result shows that full \fole is too expressive to be equivalent to \hyltl. To this end, we apply a much stronger result due to Bozzelli et al.~\cite{BozzelliMP15} showing that a certain property expressible in \kltl (\ltl with the epistemic knowledge operator~$\K$~\cite{FaginHMV95}) is not expressible in \hyctlstar, which subsumes \hyltl. 

\begin{theorem}
\label{thm_hyfoltoostrong}
There is an \fole sentence~$\phi$ that has no equivalent \hyltl sentence: For every $\hyltl$ sentence~$\phi'$ there are two sets~$T_0$ and $T_1$ of traces such that
\begin{enumerate}
	\item $\model{T_0} \not\models \phi$ and $\model{T_1} \models \phi$, but
	\item $\phi'$ cannot distinguish $T_0$ and $T_1$, i.e., either both $T_0 \models \phi'$ and $T_1 \models \phi'$ or both $T_0 \not\models \phi'$ and $T_1 \not\models \phi'$.
\end{enumerate}
\end{theorem}

\begin{proof}
Fix $\ap = \set{p}$ and consider the following property of sets~$T$ of traces over $\ap$: there is an $n > 0$ such that $p \notin t(n)$ for every $t \in T$. This property is expressible in \fole, but Bozzelli et al.~\cite{BozzelliMP15} proved that it is not expressible in \hyltl by constructing sets~$T_0, T_1$ of traces with the desired property.\footnote{Actually, they proved a stronger result showing that the property cannot expressed in \hyctlstar, which subsumes \hyltl. As the latter logic is a branching-time logic, they actually constructed Kripke structures witnessing their result. However, it is easy to show that taking the languages of traces of these Kripke structures proves our claim.}
\end{proof}

As already noted by Bozzelli et al., the underlying insight is that \hyltl cannot express requirements which relate at some point in time an unbounded number of traces. By ruling out such properties, we obtain a fragment of \fole that is equivalent to \hyltl. Intuitively, we mimic trace quantification of \hyltl by quantifying initial positions and then only allow quantification of potentially non-initial positions on the traces already quantified. Thus, such a sentence can only express properties of the bounded number of traces selected by the quantification of initial positions.

To capture this intuition, we have to introduce some notation: 
$\exists^M x.\ \phi$ is shorthand for $\exists x.\ \min(x) \wedge \phi$ and $\forall^M x.\ \phi$ is shorthand for $\forall x.\ \min(x) \rightarrow \phi$, i.e., the quantifiers~$\exists^M$ and $\forall^M$ only range over the first positions of a trace in $\model{T}$. We use these quantifiers to mimic trace quantification in \hyltl.

Furthermore, $\exists^G y \ge x.\ \phi$ is shorthand for $\exists y.\ y \ge x \wedge \phi$ and $\forall^G y \ge x.\ \phi$ is shorthand for $\forall y.\ y \ge x \rightarrow \phi$, i.e., the quantifiers $\exists^G$ and $\forall^G$ are guarded by a free variable~$x$ and range only over greater-or-equal positions on the same trace that $x$ is on. We call the free variable~$x$ the \emph{guard} of the quantifier.

We consider sentences of the form
\begin{equation}
\phi = Q_1^M x_1.\cdots Q_k^M x_k.\ Q_{1}^G y_{1}\ge x_{g_1}.\cdots Q_{\ell}^G y_{\ell} \ge x_{g_\ell}.\  \psi
\label{guardedphi}
\end{equation}
with $Q \in \set{\exists, \forall}$, where we require the sets $\set{x_1, \ldots, x_k}$ and $\set{y_{1}, \ldots, y_{\ell}}$ to be disjoint, every guard~$x_{g_j}$ to be in $\set{x_1, \ldots, x_k}$, and $\psi$ to be quantifier-free with free variables among the $\set{y_{1}, \ldots, y_{\ell}}$. We call this fragment~\hyfol. Note that the subformula starting with the quantifier~$Q^G_1$ being in prenex normal form and $\psi$ only containing the variables~$y_j$ simplifies our reasoning later on, but is not a restriction. 

\begin{theorem}\label{theo_Hyperkamp}
\hyltl and \hyfol are equally expressive. 
\end{theorem}

We prove this result by presenting effective translations between \hyltl and \hyfol (see Lemma~\ref{lemma_hyfol2hyltl} and Lemma~\ref{lemma_hyltl2hyfol}). We begin with the direction from \hyfol to \hyltl. Consider a \hyfol sentence~$\phi$ as in (\ref{guardedphi}). It quantifies $k$ traces with the quantifiers~$\exists^M$ and $\forall^M$. Every other quantification is then on one of these traces. As trace quantification is possible in \hyltl, we only have to take care of the subformula starting with the guarded quantifiers. After replacing these quantifiers by unguarded  ones, we only have to remove the equal-level predicate to obtain an \fol sentence. To this end, we merge the $k$ traces under consideration into a single one, which reduces the equal-level predicate to the equality predicate (cf.~\cite{Thomas09}). The resulting sentence is then translated into \ltl using the theorem of Gabbay et al., the merging is undone, and the quantifier prefix is added again. We show that the resulting sentence is equivalent to the original one. 

Fix a \hyfol sentence~$\phi$ as in (\ref{guardedphi}) and consider the subformula
\[\chi = Q_{1}^G y_{1}\ge x_{g_1}. \cdots Q_{\ell}^G y_{\ell} \ge x_{g_\ell}.\ \psi
\]
obtained by removing the quantification of the guards. We execute the following replacements to obtain the formula~$\chi_m$:
\begin{enumerate}
	\item Replace every guarded existential quantification~$\exists^G y_{j}\ge x_{g_j}$ by $\exists y_{j}$ and every guarded universal quantification~$\forall^G y_{j}\ge x_{g_j}$ by $\forall y_{j}$.
	\item Replace every atomic formula~$P_a(y_j)$ by $P_{(a,g_j)}(y_j)$, where $x_{g_j}$ is the guard of $y_j$.
	\item Replace every atomic formula~$\el(y_{j},y_{j'})$ by $y_j = y_{j'}$. 
\end{enumerate}
As we have removed all occurrences of the free guards, the resulting formula~$\chi_m$ is actually a sentence over the signature~$\set{<} \cup \set{ P_a \mid a \in {\ap \times \set{1, \ldots, k}}}$, i.e., an \fol sentence. 

Given a list~$(t_{1}, \ldots, t_{k})$ of traces over $\ap$, define the trace~$\merge(t_{1}, \ldots, t_{k}) = A_0 A_1 A_2 \cdots$ over ${\ap \times \set{1, \ldots, k}}$ via 
$A_n = \bigcup_{j = 1}^{k} t_{j}(n) \times \set{j} $,
i.e., we merge the $t_j$ into a single trace. 

\begin{claim}
\label{claim_fomerge}
Let  $T$ be a set of traces and let $\beta_0 \colon \set{x_1, \ldots, x_k} \rightarrow T \times \set{0}$ be a variable valuation of the guards~$x_1, \ldots, x_k$ to elements of $\model{T}$. Then, $(\model{T},
\beta_0 )\models \chi$ if, and only if, $\model{\merge(t_{1}, \ldots, t_{k})} \models \chi_m$, where $t_{j}$ is the unique trace satisfying~$\beta_0(x_{g_j}) = (t_j, 0)$. 
\end{claim}

This claim can be proven by translating a winning strategy for either player in the model checking game~\cite{Gradeletal2005}  for $(\model{T}, \chi)$ (starting with the initial variable valuation~$\beta_0$) into a winning strategy for the same player in the model checking game for $(\model{\merge(t_{1}, \ldots, t_{k})}, \chi_m )$. 

Now, we apply the theorem of Gabbay et al.~\cite{Gabbayetal80} to $\chi_m$ and obtain an \ltl formula $\chi_m'$ over $\ap \times \set{1, \ldots, k}$ that is equivalent to $\chi_m$. Let $\chi'$ be the \hyltl formula obtained from $\chi_m'$ by replacing every atomic proposition~$(a,j)$ by $a_{\pi_{j}}$, i.e., we undo the merging. The following claim is proven by a simple structural induction over $\chi_m$. 

\begin{claim}
\label{claim_ltlunmerge}
Let  $T$ be a set of traces and let $\Pi \colon \set{\pi_1, \ldots, \pi_k} \rightarrow T$ be a trace assignment. Then, $\merge(\Pi(\pi_1), \ldots, \Pi(\pi_k)) \models \chi_m'$ if, and only if, 
 $(T,\Pi) \models \chi'$. 
\end{claim}

Now, we add the quantifier prefix $Q_1 \pi_1. \cdots Q_k \pi_k.$ to $\chi'$, where $Q_j = \exists$, if $Q_j^M = \exists^M$, and $Q_j = \forall$, if $Q_j^M = \forall^M$. Call the obtained \hyltl sentence $\phi'$. 

\begin{lemma}
\label{lemma_hyfol2hyltl}
For every \hyfol sentence~$\phi$, there is a \hyltl sentence~$\phi'$ such that for every $T \subseteq (\pow{\ap})^\omega$: $\model{T} \models \phi$ if, and only if, $T \models \phi'$. 
\end{lemma}

\begin{proof}
Fix a \hyfol sentence~$\phi$ and let the $\chi$, $\chi_m$, $\chi_m'$, $\chi'$, and  $\phi'$ be as constructed as above. Let $\beta_0$ be a variable valuation as in Claim~\ref{claim_fomerge}, let the traces~$t_1, \ldots, t_k \in T$ be defined as in this claim, and let the trace assignment $\Pi$ map $\pi_j$ to $t_j$.

 Then, the following equivalences hold:
\[
(\model{T},\beta_0) \models \chi \myeq{Claim \ref{claim_fomerge}} \model{\merge(t_{1}, \ldots, t_{k})} \models \chi_m  
\myeq{by def.} \merge(t_{1}, \ldots, t_{k}) \models \chi_m' 
\myeq{Claim \ref{claim_ltlunmerge}} (T,\Pi) \models \chi'.
\]

Finally, the equivalence of $\phi$ and $\phi'$ follows from the fact that one can identify quantification of initial elements of paths in $\model{T}$ and trace quantification in $T$, as both $\phi$ and $\phi'$ have the \textit{same} quantifier prefix.
\end{proof}

It remains to consider the translation of \hyltl into \hyfol, which is straightforward, as usual.

\begin{lemma}
\label{lemma_hyltl2hyfol}
For every \hyltl sentence~$\phi$, there is a \hyfol sentence~$\phi'$ such that for every $T \subseteq (\pow{\ap})^\omega$: $T \models \phi$ if, and only if, $\model{T} \models \phi'$. 
\end{lemma}

\begin{proof}
Let $\pi_1, \ldots, \pi_k$ be the trace variables appearing in $\phi$ and fix a set~$G = \set{x_1, \ldots, x_k, x_t}$ of first-order variables, which we use as guards: the $x_j$ with $j \le k$ are identified with the trace variables and we use variables guarded by $x_t$ to model the flow of time. 
We inductively construct a formula~$\fo{\phi}$ satisfying the following invariant: For each subformula $\psi$ of $\phi$, the free variables of the formula~$\fo{\psi}$ comprise of a subset of $G$ and one additional (different!) variable, which we call the time-variable of $\fo{\psi}$. We require the time-variables of the subformulas to be fresh unless stated otherwise and also different from the guards in $G$. Intuitively, the time-variables are used to mimic the flow of time when translating a temporal operator. Formally, we define:

\begin{itemize}
	
	\item $\fo{a_{\pi_j}} = \exists^G y \ge x_j.\ \el(y,z) \wedge P_a(y)$, i.e., $z$ is the time-variable of $\fo{a_{\pi_j}}$.
	
	\item $\fo{\neg \psi_1} = \neg\fo{\psi_1}$, i.e., the time-variable is unchanged.
	
	\item $\fo{\psi_1 \vee \psi_2} = \fo{\psi_1'} \vee \fo{\psi_2}$, where we assume w.l.o.g.\ that $\fo{\psi_1}$ and $\fo{\psi_2'}$ have the same time-variable, which is also the time-variable of the disjunction.
	
	\item $\fo{\X \psi_1} = \exists^G z_1 \ge x_t.\ \succ(z,z_1) \wedge \fo{\psi_1}$, where $z_1$ is the time-variable of $\fo{\psi_1}$. Hence, $z$ is the time-variable of $\fo{\X \psi_1}$.
	
	\item $\fo{\psi_1 \U \psi_2} = \exists^G z_2 \ge x_t.\ z \le z_2 \wedge \fo{\psi_2} \wedge \forall^G z_1 \ge x_t.\ z \le z_1 < z_2 \rightarrow \fo{\psi_1}$, where $z_i$ is the time-variable of $\fo{\psi_i}$. Hence, $z$ is the time-variable of $\fo{\psi_1 \U \psi_2}$.
	
	\item $\fo{\exists \pi_j.\ \psi} = \exists^M x_j.\ \fo{\psi}$, i.e., the time-variable is unchanged.
	
	\item $\fo{\forall \pi_j.\ \psi} = \forall^M x_j.\ \fo{\psi} $, i.e., the time-variable is unchanged.
	
\end{itemize} 

Now, we define $\phi' = \exists^M x_t.\ \exists^M z.\   x_t = z \wedge \fo{\phi}$, where $z$ is the time-variable of $\fo{\phi}$. It is straightforward to show that $\phi'$ is equivalent to $\phi$. Finally, $\phi'$ can be rewritten into prenex normal form (with quantifiers~$Q^M$ and $Q^G$!) so that the outermost quantifiers bind the guards while the inner ones are guarded.
\end{proof}


\section{Conclusion and Discussion}
\label{sec_conc}

The extension from LTL to HyperLTL has fundamentally changed the
models of the logic. While a satisfiable LTL formula is guaranteed to
have an ultimately periodic model, we have shown that there is no
guarantee that a satisfiable HyperLTL formula has a model that is
finite, $\omega$-regular, or even just $\omega$-contextfree.
Characterizing the expressive power of HyperLTL is thus a formidable challenge.
Nevertheless, the results of this paper provide a first such characterization. With the definition of $\fole$ and \hyfol,
and the resulting formulation and proof of Kamp's theorem for
hyperproperties, we have established the first connection between temporal
logics for hyperproperties and first-order logic. This connection
provides a strong basis for a systematic exploration of the models
of hyperproperties.

While hyperproperties have recently received a lot of attention from a
practical perspective
(cf.~\cite{DBLP:conf/csfw/AgrawalB16,DBLP:conf/post/ClarksonFKMRS14,DBLP:conf/cav/FinkbeinerRS15}),
their logical and language-theoretic foundations are far less
understood, and it is our hope that this paper will attract more
research into this exciting area.
An important open problem is to find a non-trivial class of languages so that every satisfiable HyperLTL formula is guaranteed to be satisfied by a model from this class.
In Section~\ref{sec_models}, we have ruled out some of the obvious candidates for such a class of languages, such as the  $\omega$-regular and $\omega$-contextfree languages. The challenge remains to identify a class of languages that is rich enough for every satisfiable HyperLTL formula.

Another major open problem is to find a temporal logic that is expressively equivalent to $\fole$.
In Section~\ref{sec_fo}, we have shown that HyperLTL is less expressive than $\fole$,
by arguing that HyperLTL cannot express requirements which relate at some point in time an unbounded number of traces. Since \kltl~\cite{FaginHMV95} can express such properties, \kltl and related epistemic temporal logics are natural candidates for logics that are expressively equivalent to $\fole$. Another promising candidate is HyperLTL with past operators, motivated by the results on \hyctlstar with past~\cite{BozzelliMP15}.

\paragraph*{Acknowledgements.} We thank Markus N.\ Rabe and Leander Tentrup for fruitful discussions.


\bibliographystyle{splncs03}
\bibliography{biblio}


\newpage
\appendix
\section{Appendix}

In this appendix, we present the proof of Claim~\ref{claim_fomerge} based on model checking games for first-order logic (see, e.g., \cite{Gradeletal2005} for details on these games).

Fix a set~$T$ of traces and a formula~$\chi$ as in Claim~\ref{claim_fomerge} and recall that the model checking game~$\game(\model{T}, \chi)$ is a finite zero-sum game of perfect information between the players \emph{Verifier} and \emph{Falsifier} where Verifier's goal is to prove that the formula $\chi$ holds in $\model{T}$. We assume w.l.o.g.\ that $\chi$ is in negation normal form, i.e., negations only appear in front of atomic formulas. 

A position of $\game(\model{T}, \chi)$ consists of a subformula~$\theta$ of $\chi$ and a variable valuation~$\beta$ mapping the free variables of $\theta$ to elements of $\model{T}$'s domain. If $\theta$ is a (possibly negated) atomic formula, then the position is terminal and is winning for Verifier, if
\begin{itemize}
	
	\item $\theta = (x = y)$ and $\beta(x) = \beta(y)$,
	
	\item $\theta = (x < y)$ and $(\beta(x), \beta(y)) \in <^{\model{T}}$ (i.e., $\beta(x) = (t, n)$, and $\beta(y) = (t',n')$ with $t = t'$ and $n < n'$),
	
	\item $\theta = E(x ,y)$ and $(\beta(x), \beta(y)) \in E^{\model{T}}$ (i.e., $\beta(x) = (t, n)$, and $\beta(t',n')$ with $n = n'$), or
	
	\item $\theta = P_a(x)$ and $\beta(x) \in P_a^{\model{T}}$.

	\item $\theta = \neg(x = y)$ and $\beta(x) \neq \beta(y)$,
	
	\item $\theta = \neg(x < y)$ and $(\beta(x), \beta(y)) \notin <^{\model{T}}$,
	
	\item $\theta = \neg E(x ,y)$ and $(\beta(x), \beta(y)) \notin E^{\model{T}}$, or
	
	\item $\theta = \neg P_a(x)$ and $\beta(x) \notin P_a^{\model{T}}$.
	
\end{itemize}
Every other terminal position is winning for Falsifier.

The moves at non-terminal positions are defined as follows: 
\begin{itemize}
	\item It is Verifier's turn at $(\theta_0 \vee \theta_1, \beta)$, where she has to pick one of the successors~$(\theta_0, \beta)$ or $(\theta_1, \beta)$.

	\item It is Falsifier's turn at $(\theta_0 \wedge \theta_1, \beta)$ where he has to pick one of the successors~$(\theta_0, \beta)$ or $(\theta_1, \beta)$.
	
	\item It is Verifier's turn at $(\exists x.\ \theta, \beta)$ where she has to pick one of the successors~$(\theta, \beta[x \mapsto (t,n)])$ for every element~$(t,n)$ of $\model{T}$'s domain. Here, $\beta[x \mapsto (t,n)]$ is the variable valuation obtained from $\beta$ by adding $x$ to its domain and mapping it to $(t,n)$.

	\item It is Falsifier's turn at $(\forall x.\ \theta, \beta)$ where he has to pick one of the successors~$(\theta, \beta[x \mapsto (t,n)])$ for every element~$(t,n)$ of $\model{T}$'s domain.
\end{itemize}

A strategy~$\sigma$ for either player~$P$ is a mapping that assigns to each non-terminal position at which it is $P$'s turn a successor. A strategy~$\sigma$ is winning from a given position, if every path starting in this position that is consistent with $\sigma$ (i.e., uses the designated successor at every position of $P$) ends in a terminal position that is winning for $P$. 

It is well-known that Verifier has a winning strategy for $\game(\model{T}, \chi)$ from a position~$(\theta, \beta)$ if, and only if, $(\model{T}, \beta) \models \theta$. Due to determinacy of finite games, we also have that Falsifier has a winning strategy for $\game(\model{T}, \chi)$ from $(\theta, \beta)$ if, and only if, $(\model{T}, \beta) \not\models \theta$.

The model checking game $\game(\model{\merge(t_1, \ldots, t_k)}, \chi_m)$ for $(\model{\merge(t_1, \ldots, t_k)}, \chi_m)$ is defined analogously and the same characterization of (non-)satisfaction in terms of the existence of a winning strategy for Verifier (Falsifier) holds. 

\begin{proof}[Proof of Claim \ref{claim_fomerge}]
For the sake of readability, we denote $\game(\model{T},\chi) $  by $\game$ and accordingly $\game(\model{\merge(t_1, \ldots, t_k)}, \chi_m)$ by $\game_m$. 

Recall that we have
\[
\chi = Q_{1}^G y_{1}\ge x_{g_1}. \cdots Q_{\ell}^G y_{\ell} \ge x_{g_\ell}.\ \psi.
\]
and
\[
\chi_m = Q_{1} y_{1}. \cdots Q_{\ell} y_{\ell}.\ \psi_m
\]
where $\psi_m$ is obtained from $\psi$ by replacing every equal-level predicate~$\el$ by equality and by replacing every atomic formula~$P_a(y_j)$ by $P_{(a, g_j)}(y_j)$. In particular, the structure of the formulas (and hence the structure of the induced model checking games) is very similar, only the atomic formulas differ. Thus, we can define a mapping~$f$ from positions of $\game_m$ to positions of $\game$ as follows: a position~$(\theta_m, \beta_m)$ of $\game_m$ is mapped to $(\theta, \beta)$, where $\theta$ is the subformula of $\chi$ corresponding to $\theta_m$, and where $\beta$ is defined as follows:
$\beta(x_j) = (t_j, 0)$ and $\beta(y_j) = (t_{g_j}, \beta_m(y_j))$, i.e., we use the guard~$x_{g_j}$ of $y_j$ to determine the trace to which we map $y_j$ in $\model{T}$.

Now, to prove our claim, it suffices to show that a winning strategy for Verifier (Falsifier) in $\game$ from the position~$(\chi, \beta_0)$ can be translated into a winning strategy for Verifier (Falsifier) in $\game_m$ from $(\chi, \beta_\emptyset)$, where $\beta_\emptyset$ is the variable valuation with empty domain. Here, we only present the translation for Verifier. The translation of a winning strategy for Falsifier is analogous.

Thus, fix a winning strategy~$\sigma$ for Verifier in $\game$ from the position~$(\chi, \beta_0)$. To define the strategy~$\sigma_m$ for Verifier in $\game_m$, consider a non-terminal position~$(\theta_m, \beta_m)$ of $\game_m$ at which it is Verifier's turn, and let $f(\theta_m, \beta_m) = (\theta, \beta)$. We have to consider two cases:
\begin{enumerate}

	\item If $\theta_m $ is a disjunction, then $\theta$ is also a disjunction. In this case, we mimic the choice of $\sigma$ at $\theta$, i.e., if $\sigma$ in $\game$ picks the first (second) disjunct, then we define $\sigma_m$ to pick the first (second) disjunct as well.
	
	\item If $\theta_m = \exists y_j. \theta_m'$, then $\theta = \exists y_j. y_j \ge x_{g_j} \wedge \theta'$. Now, let $(\theta', \beta')$ be the successor of $(\theta, \beta)$ picked by  $\sigma$. Then, $\beta'(y_j)$ is on the same trace as $\beta'(x_j) = \beta(x_j)$, as otherwise Falsifier could move to the conjunct $(y_j \ge x_{g_j}, \beta')$, which is then winning for Falsifier. However, this contradicts the strategy being winning. Hence, $\beta'(y_j) = (t_{g_j}, n)$ for some $n$. We define  $\sigma_m$ so that it picks the successor $(\theta_m', \beta_m[y_j \mapsto n])$.
\end{enumerate} 

Now, consider a path from the initial position~$(\chi, \beta_\emptyset)$ of $\game_m$ to some terminal position~$(\theta_m, \beta_m)$ that is consistent with $\sigma_m$. We have to show that this terminal position is winning for Verifier. A straightforward induction shows that mapping the path pointwise to positions of $\game$ using $f$ yields a path of $\game$ from the initial position~$(\chi, \beta)$ that is consistent with $\sigma$. Hence, the terminal position~$f(\theta_m, \beta_m) = (\theta, \beta)$ reached in $\game$ is winning for Verifier. Furthermore, the property
\begin{equation}
\label{eq_beta}	\beta(y_j) = (t_{g_j}, \beta_m(y_j))
\end{equation} is satisfied for every $y_j$ by construction of the paths. 

We conclude by a case distinction over the types of (negated) atomic formulas.
\begin{itemize}
	
	\item If $\theta = (y_{j} = y_{j'})$, then we have $\beta(y_j) = \beta(y_{j'})$.
	Also, $\theta_m = (y_j = y_{j'})$ by definition. Thus, (\ref{eq_beta}) implies that $(\theta_m, \beta_m)$ is winning for Verifier.
	
	\item If $\theta = E(y_j, y_{j'})$, then we have $\beta(y_j) = (t, n)$ and  $ \beta(y_{j'}) = (t', n)$ for some $n$.  Also, $\theta_m = (y_j = y_{j'})$ by definition. Thus, (\ref{eq_beta}) implies that $(\theta_m, \beta_m)$ is winning for Verifier.
	
	\item If $\theta = (y_{j} < y_{j'})$, then we have $\beta(y_j) = (t, n)$ and  $ \beta(y_{j'}) = (t', n')$ for some $n,n'$ with $n < n'$.
	Also, $\theta_m = (y_j < y_{j'})$ by definition. Thus, (\ref{eq_beta}) implies that $(\theta_m, \beta_m)$ is winning for Verifier.
	
	\item If $\theta = P_a(y_j)$, then we have $\beta(y_j) = (t_{g_j}, n) \in P^{\model{T}}_a$, for some $n$. Also, $\theta_m = P_{(a, g_j)}(y_j)$ by definition. Then, the definition of $\model{\merge(t_1, \ldots t_k)}$ and  (\ref{eq_beta}) imply that $(\theta_m, \beta_m)$ is winning for Verifier.
	
	\item The cases of negated atomic formulas are dual.
\end{itemize} 
\end{proof}

\end{document}